\documentclass[11pt]{amsart}
\usepackage{amssymb,amsmath,graphics,wrapfig,color}
\newtheorem{theorem}{Theorem}[section]
\newtheorem{lem}[theorem]{Lemma}
\newtheorem{prop}[theorem]{Proposition}

\theoremstyle{definition}
\newtheorem{definition}[theorem]{Definition}
\newtheorem{example}[theorem]{Example}

\theoremstyle{remark}

\numberwithin{equation}{section}

\newcommand{\ts}{\mathcal{S}}
\newcommand{\tv}{\mathcal{V}}
\newcommand{\tp}{\mathcal{X}}
\newcommand{\sm}{\mathcal{S}_{\text{m}}}
\newcommand{\vm}{\mathcal{V}_{\text{m}}}
\newcommand{\mD}{\mathcal{D}}
\newcommand{\mE}{\mathcal{E}}
\newcommand{\Z}{\mathbb{Z}}
\newcommand{\mP}{\mathfrak{P}}
\DeclareMathOperator{\AVG}{AVG}
\DeclareMathOperator{\EG}{EG}
\DeclareMathOperator{\REG}{REG}

\begin{document}


\title{Measuring Political Gerrymandering}
\author{Kristopher Tapp}
\address{Department of Mathematics\\ Saint Joseph's University\\ 5600 City Avenue
         Philadelphia, PA 19131}
\email{ktapp@sju.edu}
\date{\today}
\begin{abstract} In 2016, a Wisconsin court struck down the state assembly map due to unconstitutional gerrymandering.  If this ruling is upheld by the Supreme Court's pending 2018 decision, it will be the fist successful political gerrymandering case in the history of the United States.  The \emph{efficiency gap} formula made headlines for the key role it played in this case.  Meanwhile, the mathematics is moving forward more quickly than the courts.  Even while the country awaits the Supreme Court decision, alternative versions of the efficiency gap formula have been proposed, analyzed and compared.  Since much of the relevant literature appears (or will appear) in law journals, we believe that the general math audience might find benefit in a concise self-contained overview of this application of mathematics that could have profound consequences for our democracy.
\end{abstract}

\maketitle

\section{introduction}
\emph{Partisan gerrymandering} means the manipulating of voting district boundaries for the advantage of one political party.  Although the Supreme Court has indicated that extreme partisan gerrymandering is unconstitutional, it failed to throw out the particular state maps under consideration in \emph{Davis v. Bandemer} (1986), \emph{Vieth v. Jubelirer} (2004) and \emph{LULAC v. Perry} (2006).  Justice Anthony Kennedy wrote that the Court had found ``no discernible and manageable standard for adjudicating political gerrymandering claims,''  but his opinion left the door open for future gerrymandering cases by enumerating the properties that he believed a manageable standard would require.  Motivated by Kennedy's criteria, Stephanopoulos and McGhee proposed their \emph{efficiency gap} formula to measure the degree of partisan gerrymandering in an election~\cite{McGhee1},\cite{SM}.  Their formula was one key to the plaintiffs' success in the \emph{Gill v. Whitford} (2016) case, in which a Wisconsin court struck down the state assembly map.  The case was appealed to the Supreme Court, with a decision expected in June 2018.

Meanwhile, alternative versions of the efficiency gap formula have been proposed and studied by McGhee, Nagle, Cover and others. Our purpose is to survey the mathematical (rather than legal) aspects of theses works, and provide examples, novel illustrations and a few new results to support our conclusion that one of the new alternatives is ultimately preferable to the original efficiency gap formula.

\section{A simple example}
Redistricting matters!  Image a simple state with 50 voters, 30 loyal to the red party and 20 to the blue party.  Figure~\ref{F:Examp} illustrates three possible ways to partition these voters into 5 districts.

\begin{figure}[ht!]
   \scalebox{.2}{\includegraphics{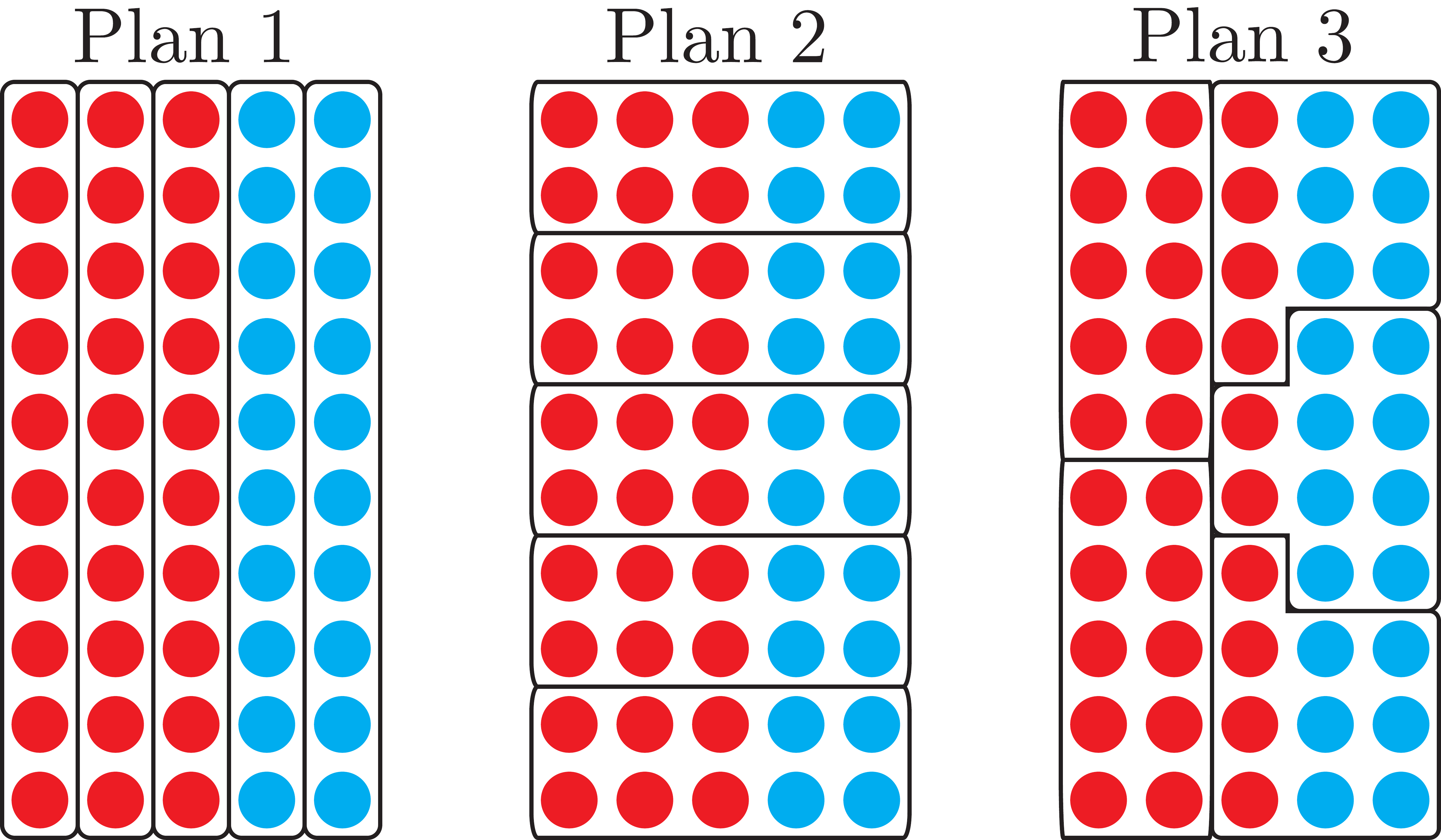}}
\caption{Three ways to divide 50 voters into 5 districts}\label{F:Examp}
   \end{figure}

Let $\tv$ denote the proportion of votes received by the red party and $\ts$ the proportion of districts (also called seats) won by the red party.  Assume full voter turnout, so that $\tv=0.6$.  Plan 1 is \emph{proportional}, which means that $\ts=\tv$.  But the red party would prefer plan 2 with $\ts=1$, while the blue party would prefer plan 3 with $\ts=0.4$.  Notice that plan 1 is the least competitive, which would make for boring election night television.

Plan 3 exhibits telltale features of a map that was gerrymandered by the blue party.  The opponent red voters were \emph{packed} into two districts where they wasted votes by having far more than the required $50\%$, and the rest of the red opponent voters were \emph{cracked} (thinly distributed) into districts where they lacked a majority, and therefore wasted their votes on losing candidates.  In general, gerrymandering involves the intertwined strategies of packing (wasting opponent votes on unnecessary super-majorities) and cracking (wasting opponent votes on losing candidates).  It's all about wasting opponent votes.  We will use this example to test each gerrymander-detection method discussed in this paper.

Plan 3 is the best that the blue party can do because in an election with $5$ districts of equal voter-turnout, the pair $(\tv,\ts)$ will lie in the interior of one of the 6 horizontal lines illustrated in Figure~\ref{F:domain} (left).  The left ends of these horizontal lines indicate that the red party needs more than $10\%$ of the votes to capture one seat, more than $20\%$ to capture two seats, etc.  The right ends indicate that the blue party has exactly these same restrictions.  These lines all lie in the parallelogram $\mP$ with vertices $\{(0,0),(.5,0),(.5,1),(1,1)\}$ illustrated in Figure~\ref{F:domain} (right).

\begin{figure}[ht!]
   \scalebox{.15}{\includegraphics{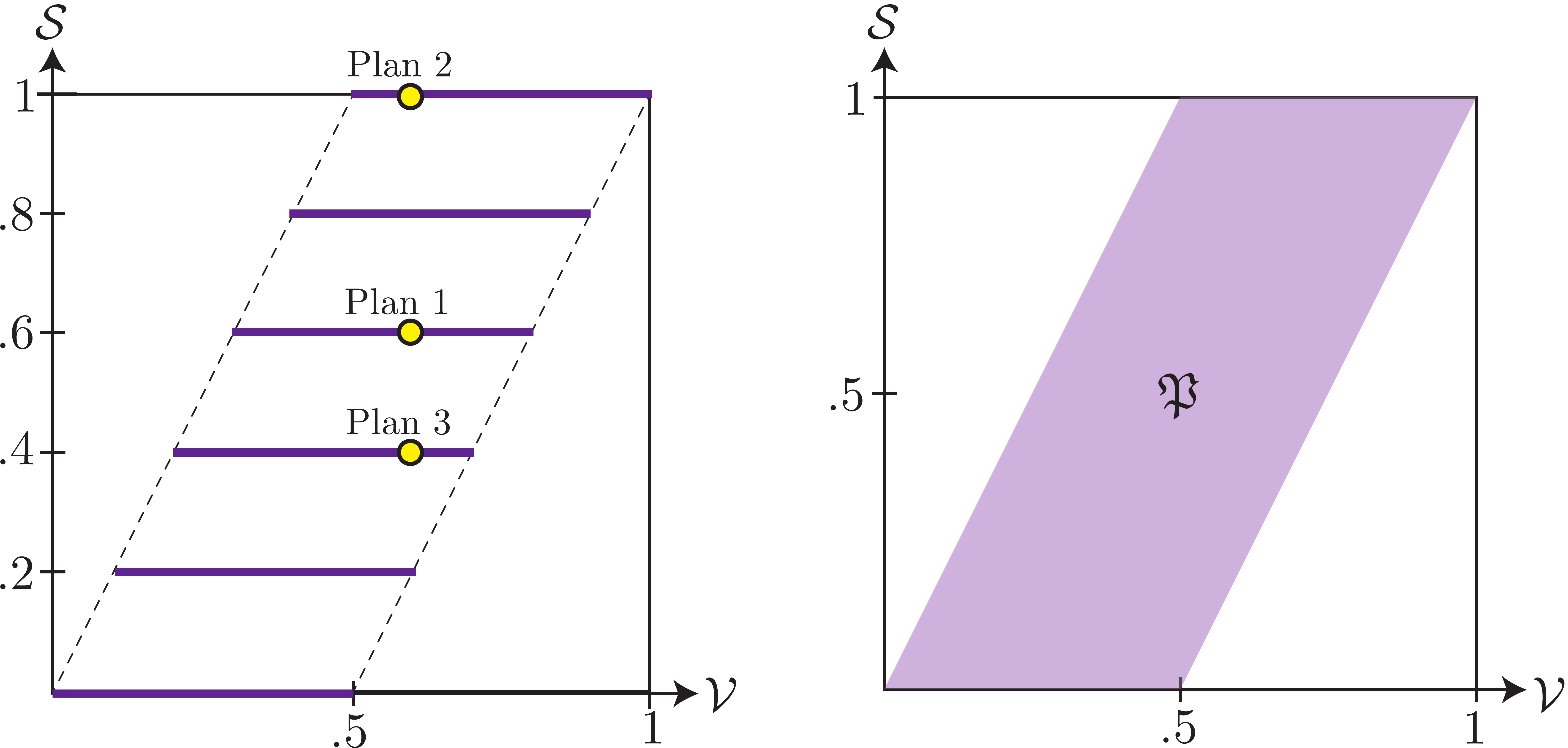}}
\caption{All possible $(\tv,\ts)$-outcomes lie in $\mP$.}\label{F:domain}
   \end{figure}

\section{Setup}
In the remainder of this paper, we consider a simple model of a state divided into $n$ districts with exactly two political parties, called $A$ (red) and $B$ (blue).  If $i\in\{1,...,n\}$ and $P\in\{A,B\}$, we define:
\begin{align*}
V^P_i  &=  \text{the number of votes cast for party $P$ in district $i$}\\
S^P_i  &=  \begin{cases} 1 &\text{if party $P$ won district $i$}\\
                          0 &\text{otherwise.}\end{cases}\\
\Gamma^P &=\left\{i\mid S^P_i=1\right\} = \text{ the set of districts won by party $P$}.
\end{align*}

Omitted superscripts and subscripts will be interpreted as summed over.  For example, $V^P$ denotes the number of votes cast for party $P$ in all districts, $V_i$ denotes the number of votes cast by both parties in district $i$, and $V$ denotes the total number of statewide votes cast.  With this convention, notice that $S^P$ is the number of seats won by party $P$ (which explains the variable name), while $S=n$.

As in the previous section, we will use the calligraphy font for the proportion of votes and seats won by party $A$, and a subscript $m$ for the marginal version of these measurments (the amount above $\frac 12$):
$$ \tv = \frac{V^A}{V}, \qquad \vm=\tv-\frac 12,  \qquad \ts = \frac{S^A}{S},\qquad \sm=\ts-\frac 12.$$
In the examples from the previous section, party A received $60\%$ of the statewide vote, so $\tv=0.6$ and $\vm=0.1$.

In this paper, we will only consider measurements and methods that detect gerrymandering purely from the district-outcome-data:
$$\mathcal{D}=\left((V^A_1,V^B_1),...,(V^A_n,V^B_n)\right).$$
This restriction forces us to ignore geometric measurements, like squared perimeter divided by area, that attempt to flag bizarrely shaped districts as gerrymandered, as well as very promising recent computer modeling work by the group \emph{Quantifying gerrymandering @Duke University}.  For example, \cite{DUKE} argues that the outcome of the Wisconsin state assembly election was extremely pro-Republican compared to a large number of simulated elections using district maps randomly sampled among the set of maps that respect geometric/geographical criteria at least as well as the actual map did.

The plaintiffs in recent gerrymandering cases didn't have to choose, but rather introduced geometric and statistical evidence in addition to evidence based on the types of measurements discussed in this paper.

District maps are legally required to have approximately equipopulus districts.  We henceforth make the slightly stronger assumption that the districts have equal voter turnout:
$$
\text{\textbf{The equal turnout hypothesis:} }V_i=V/n \text{ for each }i\in\{1,...,n\}.
$$
This hypothesis insures that $(\tv,\ts)$ lies in the interior of $\mP$ (the parallelogram in Figure~\ref{F:domain}).  In fact, the set of possible $(\tv,\ts)$-outcomes fill $\mP$ more and more densely as $V,n\rightarrow\infty$.  The right and left edges of $\mP$ represent outcomes that would only be possible if a tied district were awarded to one of the parties.  We henceforth only consider elections without any tied districts, so that $S^A+S^B=S$ (all seats are won).  

\section{Symmetry methods}
In this section, we review the basic idea of the symmetry methods that dominated the literature on gerrymander-detection through the \emph{LULAC v. Perry} (2006) Supreme Court decision.

An election result determines a single ordered pair $(\tv,\ts)\in \mP$.  But the district-outcome-data $\mathcal{D}=\left((V^A_1,V^B_1),...,(V^A_n,V^B_n)\right)$ allows one to determine what value of $(\tv,\ts)$ would have resulted had there been a uniform voter opinion shift in favor or against party $A$.  More precisely, for any $m\in\Z$, suppose that exactly $m$ voters in each district had switched from $B$ to $A$ (interpreted as vice-versa if $m$ is negative).  Define $(\tv_m,\ts_m)\in \mP$ as the outcome that would have resulted from the corresponding modified district-outcome-data
$$ \mathcal{D}_m=\left((V^A_1+m,V^B_1-m),...,(V^A_n+m,V^B_n-m)\right),$$
with voter-counts less than zero interpreted as zero and voter-counts larger than $V_i$ interpreted as $V_i$.

There are more sophisticated and more reasonable ways to model a ``uniform'' shift in voter opinion, say in favor of party A.  For example, voters could be flipped one at a time from party B to A, with all party B voters in all district being equally likely to be the next to flip (see~\cite{NagleP}); this method avoids the issue of voter counts going below zero and above $V_i$.  But our simple additive model suffices to demonstrate the key ideas here.

The set $\{(\tv_m,\ts_m)\mid m\in\Z\}$  is a \emph{simulated seats-votes curve}.  For each possible value of $\tv$, it shows the portion of seats that party $A$ would have won if a uniform shift in voter opinions had caused them to received that fraction of the votes.

\begin{figure}[ht!]
   \scalebox{.14}{\includegraphics{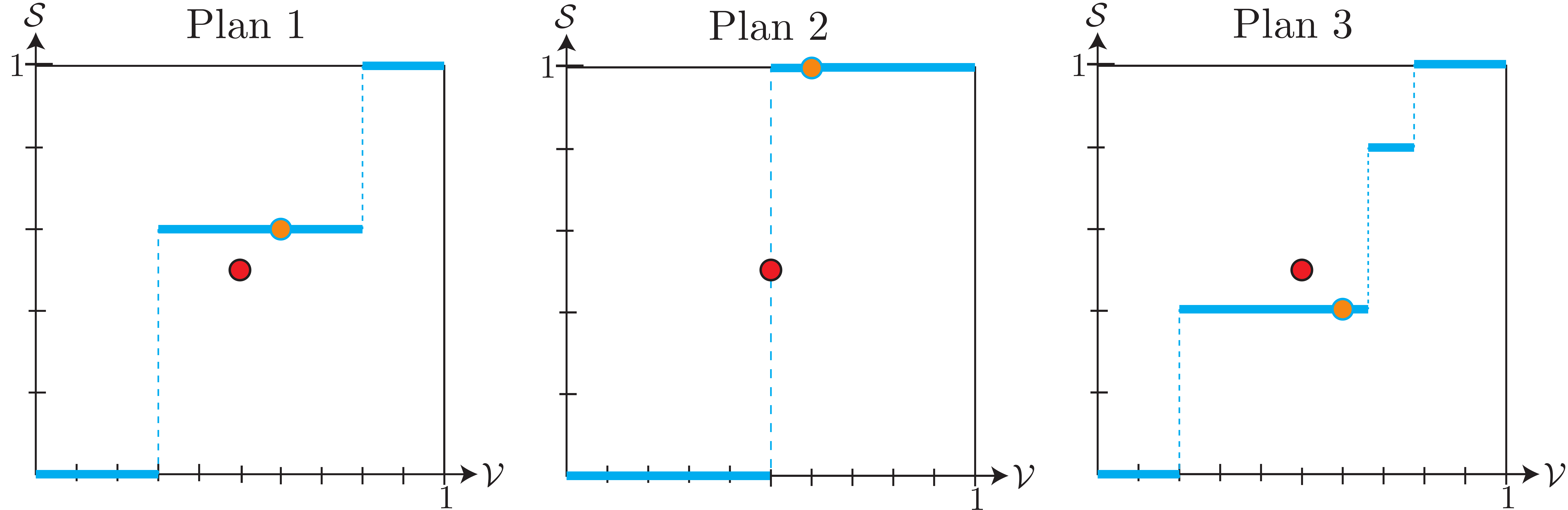}}
\caption{Seats-votes curves for the plans from Figure~\ref{F:Examp}}\label{F:SV}
   \end{figure}

Figure~\ref{F:SV} illustrates the simulated seats-votes curve for the district plans from Figure~\ref{F:Examp} (with each circle in Figure~\ref{F:Examp} representing 1000 voters rather than 1 voter).  Only plan 2 is fair in the sense that its curve is symmetric about the point $(.5,.5)$.  In the actual election, party A won $100\%$ of the seats with $60\%$ of the votes, but party $B$ would have received the same reward -- $100\%$ of the seats -- if they had been the party who received $60\%$ of the votes.  This district map treats the two parties equally.

Plan 1 exhibits bias in favor of party A (red), while plan 3 exhibits bias in favor of party B (blue).  There are several precise methods in the literature used to measure bias -- how much each graph fails to be symmetric about the point $(.5,.5)$.  The simplest is just the graph's height above $(.5,.5)$, which equals $0.1$ of Plan 1 and equals $-0.1$ for Plan 3.  So a simulated voter shift to $\tv=.5$ causes party A receives $60\%$ of the seats with plan 1 and $40\%$ of the seats with plan 3.  There is some some legal weight behind the principle that a party receiving more than half the votes should receive at least half of the seats.  The actual outcome of plan 3 violated this principle, as did the simulated outcome of plan 1.

In addition to bias, another commonly discussed measurement is the \emph{responsiveness} of a seats-votes curve, which usually means its average rate of change over an interval like say $\tv\in[.45,.55]$.  High responsiveness means that the districts are more competitive, so that small changes in voter preference have larger effects on the number of seats obtained, which is usually thought of as a desirable property.

For example, ongoing legal challenges to the Pennsylvania congressional map after the 2012 election are based not just on its bias in favor of Republicans, but also on its low responsiveness.  The Republicans won $13$ of the $18$ seats with only about $49\%$ of the statewide vote, and the simulated seats-votes curve was nearly constant on $\tv\in[.4,.6]$, which means that the Democrats could not improve their unfair situation even by winning more votes.

Our short and simple discussion in this section doesn't do justice to the abundance of literature on symmetry measurements of seats-votes curves.  There are piles of papers containing more sophisticated ways to model uniform voter shifts, construct seats-votes curves, measure their deviation from being symmetric about $(.5,.5)$, perform statistical analyses, and anchor the measurements to legal principles.  For an expanded view, we recommend~\cite{NagleP},\cite{GK},\cite{GrK} and references therein.

These symmetry-based measurements of gerrymandering failed to impress a majority of the Supreme Court justices in cases up to and including \emph{LULAC v. Perry} (2006), partly because they are rooted in speculative and somewhat arbitrary counterfactual simulations.  This prompted the invention of the efficiency gap formula, which measures the degree of gerrymandering based on the counting of wasted votes.

\section{The efficiency gap}
As discussed in Section 2, gerrymandering boils down to forcing voters in the opponent party to waste votes, so a natural fairness principle is to require that the two parties waste about the same number of votes.  There are two types of wasted votes: ``losing votes'' cast for a losing candidate, and ``excess votes'' above the $50\%$ required to win a district.  So the number of votes wasted by party $P$ in district $i$ equals:
\begin{equation}\label{E:Waste}
W^P_i = \begin{cases} V_i^P &\text{ if party $P$ lost district $i$ (losing votes)} \\
                      \left(V_i^P - \frac{V_i}{2}\right)&\text{ if party $P$ won district $i$ (excess votes)}
         \end{cases}
\end{equation}

Recall the convention that omitted superscripts and subscripts are interpreted as summed over, so $W^P$ equals the total number of votes wasted by party $P$ in all districts.  McGhee defined the \emph{efficiency gap} in~\cite{McGhee1} as
\begin{equation}\label{E:EG}\EG = \frac{W^B-W^A}{V},\end{equation}
which is just the difference in the number of votes wasted by the two parties, divided by the total number of voters.

The goal of gerrymandering is to waste more votes from the opponent party than from one's own party, so $\EG$ is designed to measure the extent to which this occurred.  If $\EG$ is positive (party $B$ wasted more votes), then the efficiency gap is evidence that party $A$ manipulated the district boundaries for political gain.  Similarly, a negative efficiency gap is evidence that party $B$ manipulated the map.

Figure~\ref{F:EG} illustrates $\EG$ for the example maps from Section 2.  All three plans have $|\EG|>0.08$, which Stephanopoulos and McGhee proposed as an indicator of gerrymandering in state assembly elections.

\begin{figure}[ht!]
   \scalebox{.2}{\includegraphics{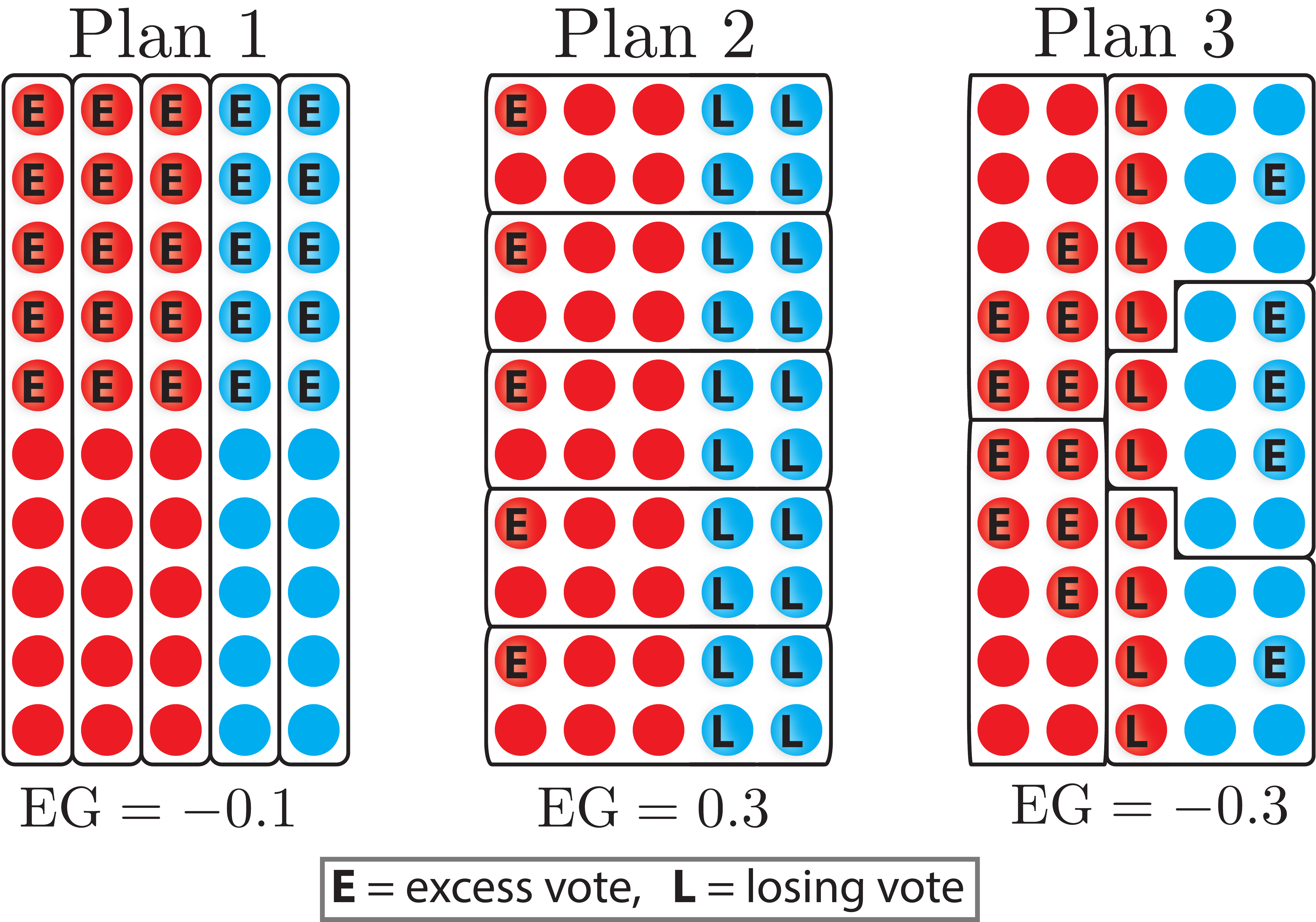}}
\caption{The efficiency gap of the plans from Figure~\ref{F:Examp}}\label{F:EG}
   \end{figure}

The $\EG$ represents a fairness principle that is sometime inconsistent with the symmetry principle of the previous section; for example, plan 2 has a symmetric seats-votes curve, but yet is rated by $\EG$ as highly gerrymandered by the red party.

Notice that $-0.5\leq \EG\leq 0.5$ because in each district half of all votes are wasted, and the extreme cases occur when all of the statewide wasted votes come from a single party.

A couple of weaknesses of $\EG$ are apparent immediately from Equation~\ref{E:EG}, but the plaintiffs in \emph{Gill v. Whitford} successfully countered arguments that the defence levied based upon these weaknesses:
\begin{itemize}
\item $\EG$ depends on the election outcome (unlike compactness measurements that depend only on the map), and is volatile in competitive races.  For example, if all districts are highly competitive and a single party happens by chance to win all of the districts, then $\EG$ would provide strong evidence that the winning party manipulated the map.  To counter this complaint, the plaintiffs showed that Wisconsin's high $\EG$ persisted in computer-simulated elections with random swings.
\item Demographic factors can cause $\EG$ to be high.  For example, Democrats tend to be packed into cities where they waste votes by having far more than the majority needed to elect the Democratic candidate.  Lawyers for the defense argued that Wisconsin's high $\EG$ is explained by demographics rather than manipulated district boundaries.  The plantiffs counted with computer simulations showing that the observed $\EG$ is high compared to the average $\EG$ of simulated elections in large numbers of random computer-generated district maps~\cite{JC}.
\end{itemize}

McGhee made the key observation that $\EG$ depends on much less information than its definition suggests:
\begin{lem}[McGhee~\cite{McGhee1}]\label{L:EG} $\EG = \sm - 2\cdot\vm.$
\end{lem}
For example in 2016, the Republicans held $65\%$ of the state assembly seats in Wisconsin ($\sm=0.15$) despite receiving only $52\%$ of the statewide vote ($\vm=0.02$).  Assuming equal turnout, this is the only information needed to compute that $\EG = 0.11$ (here $A$=Republicans and $B$=Democrats).  It might be surprising that a district-by-district tally of wasted votes is not required.  The lemma also gives a faster way to calculate and understand the results in Figure~\ref{F:EG} without the need to tally wasted votes.

\begin{proof}
Equation~\ref{E:Waste} becomes:
\begin{equation}\label{E:21}W_i^P = V_i^P - S_i^P\cdot\frac{V_i}{2} =  V_i^P - S_i^P\cdot\frac{V}{2S},\end{equation}
so the total number of votes wasted by party $P$ equals:
$$W^P = \sum_{i=1}^n W^P_i = V^P - S^P\cdot\frac{V}{2S}.$$
Therefore:
\begin{align*}\label{EGBD}
  \EG =\frac{W^B-W^A}{V}  = \frac{V^B-V^A}{V} - \frac 12\cdot \frac{S^B-S^A}{S}
                         = -2 \vm + \sm.
\end{align*}
\end{proof}

The fairness principle that $\EG$ should be small becomes:
\begin{equation}\label{E:EGsmall}\EG\cong 0 \,\,\,\iff \,\,\,\sm\cong 2\cdot \vm,\end{equation}
which conflicts with the principle of proportionality ($\sm\cong\vm$).  For example, if a party wins $60\%$ of the statewide vote, proportionality requires them to win about $60\%$ of the seats, but achieving a zero efficiency gap requires them to win $70\%$ of the seats.  Thus, proportionality is replaced with a \emph{winner's double bonus} principle: the winner deserves a seat margin equal to twice the vote margin.

The plaintiffs in \emph{Gill v. Whitford} argued that this double bonus is normative because the factor $2$ matches historical data (see \cite{Tufte} and Figure 5 of~\cite{SM2}) and because the principle is derived from the canonical activity of equating wasted votes (although the counting and equating of wasted votes here involved some arbitrary decisions that we will soon discuss).  They might also have mentioned that $\EG=0$ is a visually appealing centerline of $\mP$; More precisely, Figure~\ref{F:two_things} shows that for a given value of $\ts$, the choice of $\tv$ that makes the efficiency gap vanish is half way between its allowable extremes.  The double bonus ``looks canonical'' because the edges of $\mP$ also have slope $2$.

\begin{figure}[ht!]
   \scalebox{.17}{\includegraphics{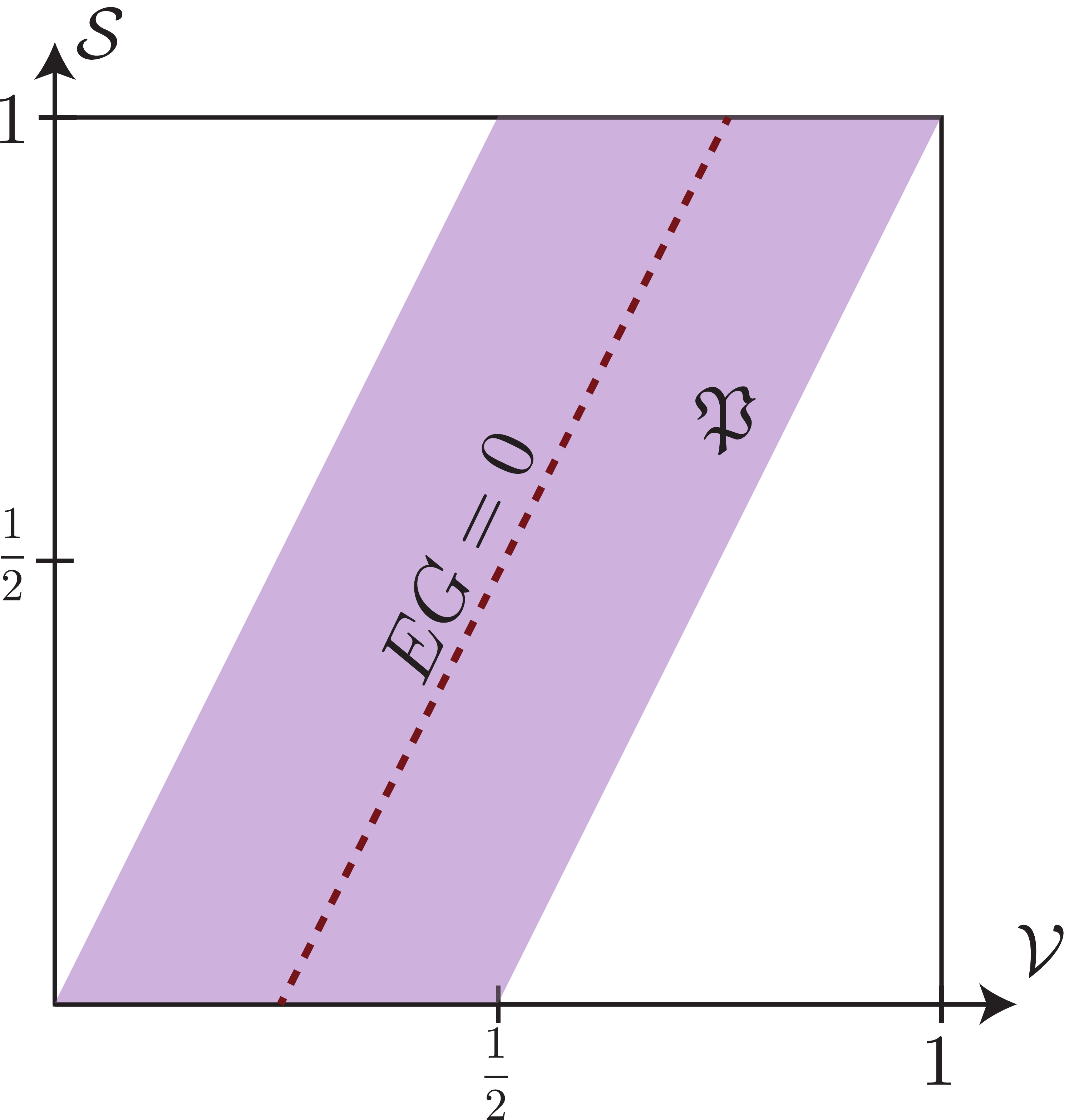}}
\caption{$\EG=0$ is a visually natural centerline of $\mP$.}\label{F:two_things}
   \end{figure}

On the other hand, the double bonus does not seem to emerge from any natural probability model.  The most naive non-state-specific model would assume the parties are uniformly distributed through the state.  There are various ways to incorporate random noise into such a model, but with a large state population, such models predict that the majority party almost always takes all the seats.

Lemma~\ref{L:EG} reveals the efficiency gap's most serious deficiency: when one party has a sufficiently large majority, the efficiency gap is guaranteed to indicate that the minority party manipulated the map, regardless of how the map was drawn.  More precisely, when party A has over $75\%$ of the statewide vote ($\vm>.25$), Lemma~\ref{L:EG} implies that $\EG<0$, which indicates that party $B$ manipulated the map.  In fact:
\begin{equation}\label{E:girraf}\vm>.25+x\Rightarrow \EG<-2x.\end{equation}
This limitation renders the efficiency gap useless in lopsided elections.  

McGhee and Stephanopoulos observed that in the past several decades there have been almost no congressional or state house elections with $\vm>.25$~\cite{SM2}.  Nonetheless, we believe that a robust mathematical formula should correctly handle extreme cases.

Before fixing this deficiency with a better formula, we will highlight more things that $\EG$ gets (sort of) right.

\section{The efficiency gap and competitiveness}
\begin{wrapfigure}{r}{0cm}
   \scalebox{.2}{\includegraphics{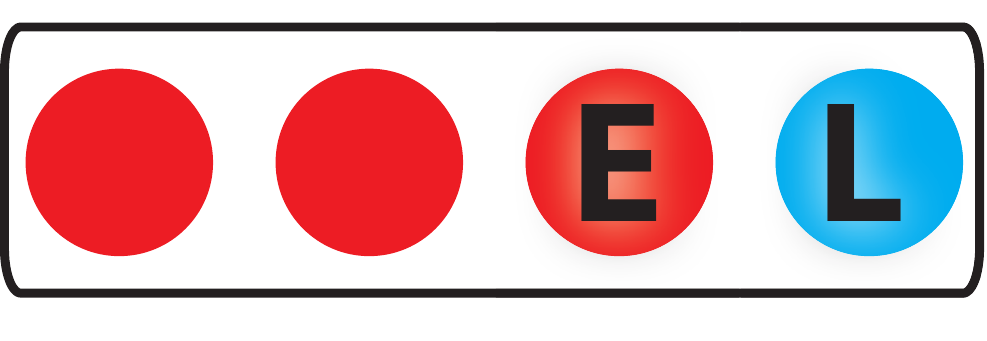}}
   \end{wrapfigure}
Does the $\EG\cong 0$ principle prevent a partisan map-making team from packing and cracking its opponents?  The answer is yes, provided one is willing to precisely define a ``packed district'' as a district won with more than $75\%$ of the vote, and a ``cracked district'' as a district lost with between $25\%$ and $50\%$ of the vote.  To understand these cut-off values here, notice that any district in which the winning party has exactly $75\%$ of the vote is neutral - such a district contributes zero to the efficiency gap because the wasted votes are evenly split between the parties, as shown in the figure on the right.

A district won with more than $75\%$ of the vote contributes evidence that the losing party manipulated the map (by packing the winning party), while a district won with less than $75\%$ of the vote contributes evidence that the winning party manipulated the map (by cracking the losing party).

This discussion can be reframed in terms of competitiveness by defining:
$$C_i = \frac{|V_i^A - V_i^B|}{V_i},\,\,\,\,\, C^P = \AVG\left\{C_i\mid i\in\Gamma^P\right\},\,\,\,\,\, C = \AVG\left\{C_i\mid 1\leq i\leq n\right\}.$$
That is, $C_i\in[0,1]$ denotes the competitiveness of district $i$, defined as the proportion of the vote by which the district was won, $C^P$ denotes the average competitiveness of $P$-won districts, and $C$ denote the average competitiveness of all districts.  Notice that $C_i=.5$ for a neutral district that contributes zero to the efficiency gap.  The efficiency gap penalizes the winner of a competitive district ($0<C_i<.5$) for cracking the loser, and it penalizes the loser of a non-competitive district ($.5<C_i<1$) for packing the winner.  In fact, the additive contribution to $\EG$ from district $i$ equals $\pm\left(C_i-\frac 12\right)$, which lead Cover to observe:

\begin{prop}[Cover~\cite{Cover}]\label{P:cover}The efficiency gap is the seat-share-weighted difference of the average of the amounts by which the competitiveness of districts won by parties $A$ and $B$ differs from $\frac 12$:
$$\EG=\underbrace{\left(\frac 12 - \sm \right)}_{\frac{S^B}{S}}
 \left(C^B-\frac 12\right) - \underbrace{\left(\frac 12 + \sm\right)}_{\frac{S^A}{S}}\left(C^A-\frac 12\right).$$
\end{prop}
\begin{proof} The efficiency gap equals the average waste-gap of the districts:
$$ \EG  = \frac{W^B-W^A}{V}=\frac{1}{n}\sum_{i=1}^n\frac{W^B_i-W^A_i}{V_i}.$$
The waste-gap in a single district is linearly related to its competitiveness:
$$\frac{W^B_i-W^A_i}{V_i} =\begin{cases} C_i-\frac 12 &\text{ if } i\in\Gamma^B \\ -C_i+\frac 12 &\text{ if } i\in\Gamma^A. \end{cases}$$
Thus,
\begin{align*}
\EG & = \frac{1}{n}\left(\sum_{i\in\Gamma^B}\left(C_i-\frac 12\right)
        -  \sum_{i\in\Gamma^A}\left(C_i-\frac 12\right)\right) \\
   & = \left(\frac{S^B}{S}\right)\cdot\AVG\left\{C_i-\frac 12\mid i\in\Gamma^B\right\}
      - \left(\frac{S^A}{S}\right)\cdot\AVG\left\{C_i-\frac 12\mid i\in\Gamma^A\right\},
\end{align*}
from which the result follows.
\end{proof}

The \emph{Gill v. Whitford} majority was swayed by evidence that Democrat-won districts were far less competitive than Republican-won districts, indicating that Republican map-makers packed Democrats into safe districts.  They regarded this \emph{differential competitiveness} evidence as independent from the efficiency gap evidence.  But Proposition~\ref{P:cover} shows that the efficiency gap is (sort of) related to differential competitiveness.  The simplest relationship is:
\begin{equation}\label{E:eg_comp}\sm=0 \Rightarrow \EG=\frac 12(C^B-C^A),\end{equation}
so parties who win equal numbers of seats are required to win equally competitive districts on average.  The general relationship is more complicated because $\EG$ measures \emph{weighted} differential competitiveness.  A picture helps.  For any fixed value of $\sm$, the equation $\EG=0$ is graphed in the $C^AC^B$-plane as the line through $(.5,.5)$ with
$$\text{slope} = \frac{.5+\sm}{.5-\sm} = \frac{S^A}{S^B}.$$

This slope, coming from the weighting, helps the party who won more seats, as can be seen in the graphs for plans 1 and 3 (from Section 2) pictured in Figure~\ref{F:competitive}.  In plan 1, all districts were won unanimously ($C^A=C^B=1$), but even though the parties won equally competitive districts, the score $\EG=-0.1$ provides evidence that the minority blue party manipulated the map.  In plan 3, party A (red) won two unanimous districts ($C^A=1$) while party B (blue) won three fairly competitive districts ($C^B=1/3$).  The weighting helps the winning blue party a bit but not enough; the efficiency gap turned out negative enough to indicate that they manipulated the map.

\begin{figure}[ht!]
   \scalebox{.15}{\includegraphics{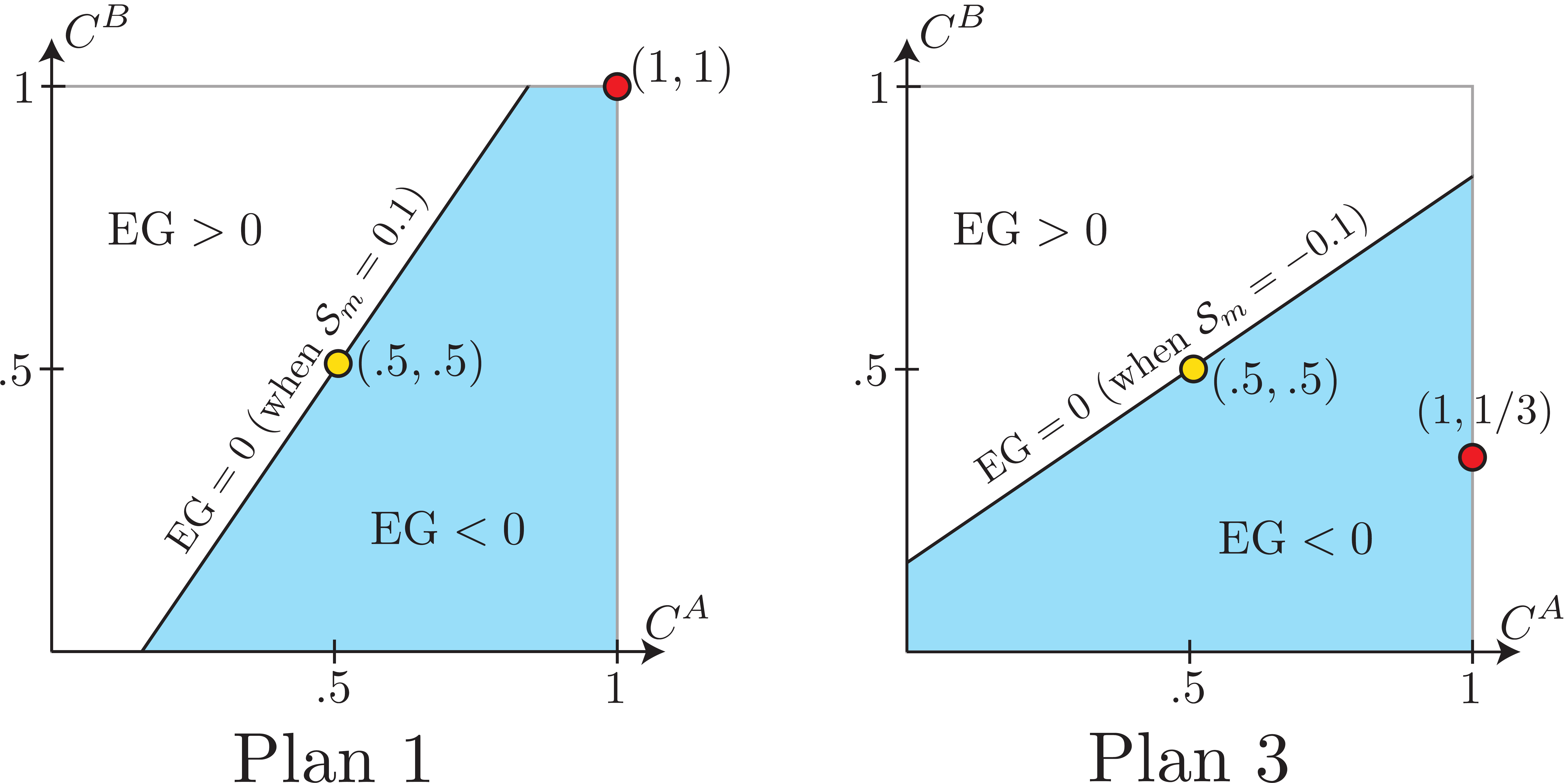}}
\caption{The $\EG$ visualized as a weighted differential competitiveness in plans 1 and 3.}\label{F:competitive}
   \end{figure}

\section{The weighted efficiency gap}
The dissenting judge in \emph{Gill v. Whitford} criticized how the $\EG$ formula counts excess votes.  An ``excess vote'' is supposed to mean a vote beyond what's needed to win a district.  If a party wins a district with $60\%$ of the vote, the $\EG$ formula counts $10\%$ of their votes as wasted.  But shouldn't $20\%$ of their votes count as wasted, since winning really only required more votes than the $40\%$ received by the losing party?  If we agree with this judge that an excess vote should mean a vote beyond the number received by the losing party, then we must alter the formula to double all excess vote counts.  This is the $\lambda=2$ case of Nagle's \emph{weighted efficiency gap} formula, which counts excess votes with an arbitrary weight $\lambda\in\mathbb{R}^+$:
\begin{align*}\label{E:Waste}
W^P_i(\lambda) & = \begin{cases} V_i^P &\text{ if $i\notin\Gamma^P$ (losing votes)} \\
                      \lambda\cdot\left(V_i^P - \frac{V_i}{2}\right)&\text{ if $i\in\Gamma^P$ (excess votes)},\end{cases}\\
   \EG_\lambda & = \frac{W^B(\lambda)-W^A(\lambda)}{V}.
\end{align*}

In our opinion, $\lambda\in\{1,2\}$ are the only natural cases, but there is no harm in allowing an arbitrary weight.  The weighted version of Lemma~\ref{L:EG} becomes:
\begin{lem}[Nagle~\cite{Nagle}]\label{L:geneg}
$\EG_\lambda = \sm - \left(1+\lambda\right)\cdot\vm$.
\end{lem}
\begin{proof}
The above weighted definition of wasted votes becomes:
$$W_i^P(\lambda) = V_i^P - S_i^P\cdot\left(\lambda\frac{V}{2S}+(1-\lambda)V^P_i\right),$$
so the total number of votes wasted by party $P$ equals:
\begin{equation}\label{E:waste_lambda_total}W^P(\lambda) = \sum_{i=1}^n W^P_i(\lambda) = V^P - \lambda S^P\cdot\frac{V}{2S}+(\lambda-1)\cdot\sum_{i\in\Gamma^P} V_i^P.\end{equation}
The contributions from the first two terms simplify as in the $\lambda=1$ special case, giving:
\begin{align*}
  \EG_\lambda = -2\vm +\lambda\cdot\sm +(\lambda-1)\cdot\left(\tp(B,B)-\tp(A,A)\right),
\end{align*}
where $\tp(X,Y)=\frac{1}{V}\cdot\sum_{i\in\Gamma^Y} \left(V_i^X\right)$ is the proportion of voters with these two properties: voting for party $X$ and living in a district won by party $Y$.  Using the relations:
$$\tp(A,A)+\tp(B,A) = \frac 12 + \sm,\quad \tp(B,A)+\tp(B,B) = \frac 12 - \vm,$$
we see that
\begin{equation}\label{E:doggie}\tp(B,B)-\tp(A,A) = -\vm - \sm,\end{equation}
which completes the proof.
\end{proof}
In particular, $$\EG_2 = \sm - 3\cdot\vm,$$ which is a \emph{winner's triple bonus}.  Proportionality advocates prefer to step in the other direction towards the choice $\lambda=0$, which corresponds to counting only losing votes: $$\EG_0 = \sm - \vm.$$  But legal arguments based on the principal $\EG_0\cong 0$ would not hold up because the Supreme Court has repeatedly ruled that political parties do not have a right to proportional representation.

A district with competitiveness $C_i = \frac{1}{\lambda+1}$ is neutral -- it contributes zero to the efficiency gap.  Although this cutoff value now depends on $\lambda$, the logic remains the same: the $\EG_\lambda$ measurement penalizes the winner of a competitive district ($0<C_i<\frac{1}{\lambda+1}$) for cracking the loser, and it penalizes the loser of a non-competitive district ($\frac{1}{\lambda+1}<C_i<1$) for packing the winner.  In fact, the weighted version of Proposition~\ref{P:cover} is:
\begin{prop}[Cover~\cite{Cover}]
$$\EG=\left(\frac 12 - \sm \right)
 \left(\frac{\lambda+1}{2}\cdot C^B-\frac 12\right) - \left(\frac 12 + \sm\right)\left(\frac{\lambda+1}{2}\cdot C^A-\frac 12\right).$$
\end{prop}
\begin{proof} The waste-gap in a single district is:
$$\frac{W^B_i-W^A_i}{V_i} =\begin{cases} \frac{\lambda+1}{2}C_i-\frac 12 &\text{ if } i\in\Gamma^B \\ -\frac{\lambda+1}{2}C_i+\frac 12 &\text{ if } i\in\Gamma^A. \end{cases}$$
The results now follow as in the proof of Proposition~\ref{P:cover}.
\end{proof}

For any fixed values of $\lambda$ and $\sm$, the equation $\EG_\lambda=0$ is graphed in the $C^AC^B$-plane as the line through $\left(\frac{1}{\lambda+1},\frac{1}{\lambda+1}\right)$ with slope=$\frac{0.5+\sm}{0.5-\sm}=\frac{S^A}{S^B}$.  In particular, Equation~\ref{E:eg_comp} remains true for arbitrary values of $\lambda$.

Several authors consider $\EG_1$'s cutoff value of $C_i=1/2$ to be a flaw, complaining that it ``fetishizes three-to-one landslide districts''\cite{Duchin}.  We agree that $\EG_2$'s cutoff value of $C_i=1/3$ seems to be more reasonable -- a district should probably count as packed if it's won with between $66\%$ and $75\%$ of the vote, but we acknowledge that there's no canonical choice for the cutoff.  \emph{Some} cutoff is necessary for any formula that is additive over the districts and attempts to penalize both packing (losing by too much) and cracking (winning by too little).

In summary, $\EG_2$ yields a more reasonable competitiveness cutoff than $\EG_1$, and its ``winner's triple bonus'' might appeal to advocates of competitive elections, but it exacerbates the main problem: $\EG_1$ is worthless for elections in which one party received more than $75\%$ of the vote, while $\EG_2$ is worthless above $66\%$.


\section{The relative efficiency gap}
It is possible to solve all of these problems at once with an elegant \emph{relative} version of the efficiency gap formula.
This idea was first proposed by Nagle in~\cite{Nagle}.  We will call it the \emph{relative efficiency gap}:
$$\REG_\lambda = \frac{W^B(\lambda)}{V^B}-\frac{W^A(\lambda)}{V^A}\in[-1,1].$$
It measures the difference between the \emph{proportions} of their votes that the two parties wasted.  The idea is to require the parties to waste about the same proportion of their votes rather than the same number of votes.  For example, it is compelling to regard Plan 1 of Figure~\ref{F:Examp} as fair because the parties waste the same proportion of their votes (even though they don't waste the same number of votes).

As before, we allow $\lambda$ to be an arbitrary positive constant, even though we think $\lambda\in\{1,2\}$ are the only important cases.

Nagle called $\EG_\lambda$ ``party-centric'' and $\REG_\lambda$ ``voter-centric.''  Making the efficiency gap small ``equalizes the aggregate harm done to a party," whereas making the relative efficiency gap small ``equalizes the average effectiveness of voters of like mind.''  In other words, $\REG_\lambda\cong 0$ means that a randomly selected voter from party A is just as likely to have wasted his/her vote as a randomly selected voter from party B.  This distinction is legally relevant because the Constitution grants rights to individuals not parties.

The global formula for $\REG_\lambda$ depends not only on $\vm$ and $\sm$, but also on the competitiveness measurement $C$ defined in Section 6:
\begin{prop}[Cover~\cite{Cover}]
$$\REG_\lambda = \frac{\sm -\left(\lambda + (1- \lambda)C\right)\vm }{2\left(\frac 12 +\vm\right)\left(\frac 12 - \vm\right)}$$
\end{prop}
\begin{proof}
Equation~\ref{E:waste_lambda_total} gives:
$$\frac{W^P(\lambda)}{V^P} = 1-\frac{\lambda}{2}\cdot\frac{S^P/S}{V^P/V} + (\lambda - 1)\frac{\tp(P,P)}{V^P/V},$$
with $\tp$ defined as the proof of Lemma~\ref{L:geneg}.  Subtracting gives:
$$\REG_\lambda = \frac{\lambda}{2}\left(\frac{\sm-\vm}{\left(\frac 12 +\vm\right)\left(\frac 12 - \vm\right)}\right) + (\lambda-1)\cdot\frac{\color{red}\left(\frac 12 + \vm\right)\cdot\tp(B,B)-\left(\frac 12 - \vm\right)\cdot\tp(A,A)\color{black}}{\left(\frac 12 +\vm\right)\left(\frac 12 - \vm\right)}.$$
The {\color{red} red numerator} above equals:
$$\frac 12 (\underbrace{\tp(B,B)-\tp(A,A)}_{=-\vm-\sm\text{ by Eq.~\ref{E:doggie}}}) + \vm\cdot (\underbrace{\tp(A,A)+\tp(B,B)}_{=\frac 12(C+1)}),$$
where the second underscored equality recognizes that the proportion of all votes that were cast for winning candidates depends linearly on the competitiveness.  Making these substitution and simplifying completes the proof.
\end{proof}

We first consider the case $\lambda=1$, in which the dependence on $C$ disappears:
\begin{equation}\label{reg2}\REG_1
   = \frac 12\left(\frac{\sm-\vm}{\left(\frac 12 +\vm\right)\left(\frac 12 - \vm\right)}\right).\end{equation}
Figure~\ref{F:graph} shows the graph of $\REG_1$ over the domain $\mP$.  The cyan line in the figure shows the slice $\vm=0$, along which $\REG_1$ is a linear function of $\sm$.  The green lines in the figure show that
$$\REG_1=0\iff \vm=\sm,$$
so the $\REG_1\cong 0$ principal is consistent with proportionality.

Notice that $\REG_1$ is defined on all of the boundary of $\mP$ except the two points $(\vm,\sm)=\pm(.5,.5)$ illustrated in yellow.  What is the limit of $\REG_1$ as $(\vm,\sm)$ approaches either of these two points along a line in $\mP$?  It depends on the choice of the line.  All values between $-.5$ and $.5$ can be obtained as limits (including the value $0$ along the green line).

\begin{figure}[ht!]
   \scalebox{.20}{\includegraphics{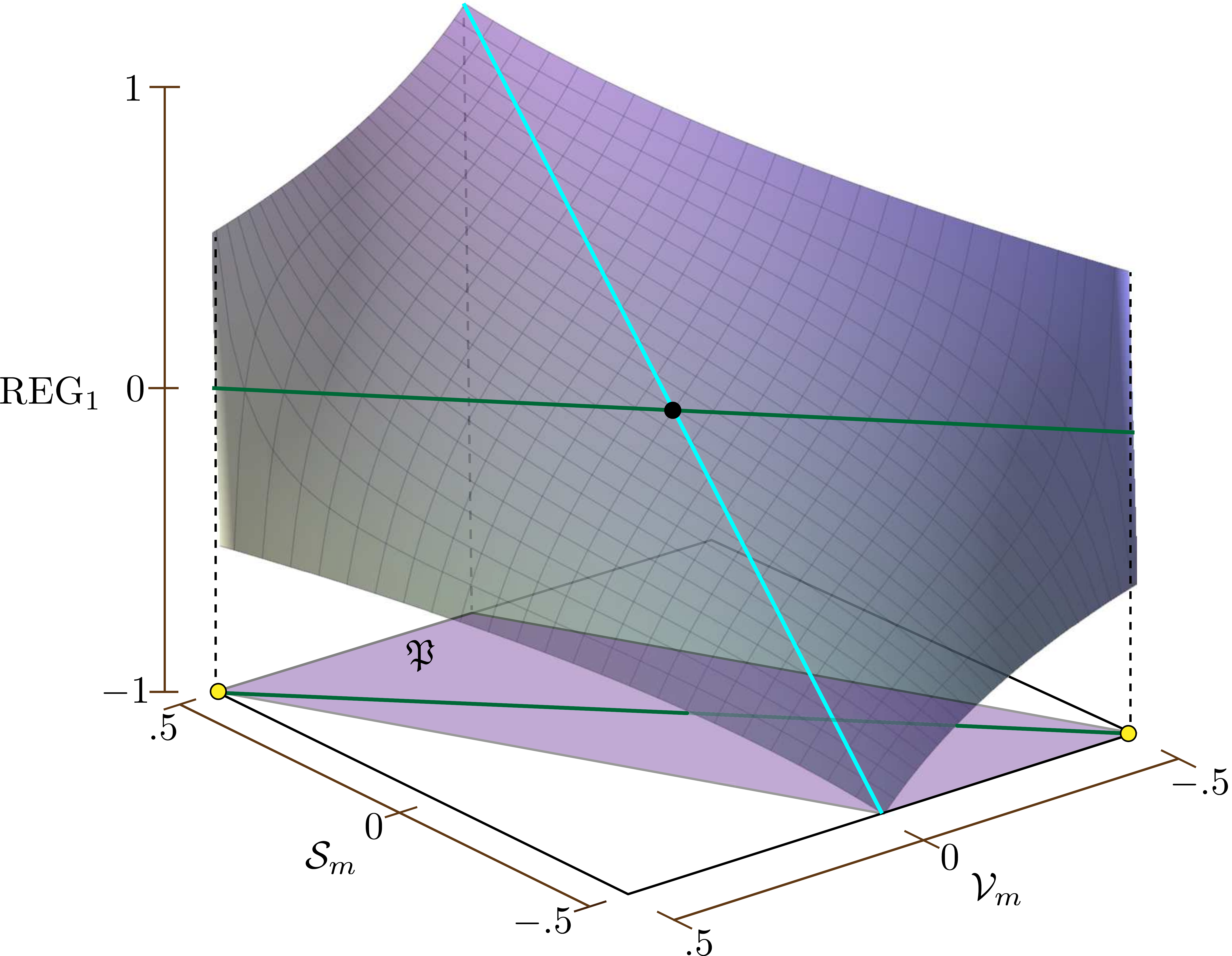}}
\caption{the graph of $\REG_1$ over the domain $\mP$}\label{F:graph}
   \end{figure}

We saw in Equation~\ref{E:girraf} that $\EG$ is useless for lopsided elections, but $\REG_1$ fares much better in this regard, provided the number of districts is reasonably large (Wisconsin has $99$ state assembly districts).  A tiny party without enough votes to capture a single seat would be guaranteed to waste all of their votes regardless of the map. But as soon as both parties have enough votes to capture a seat or two each, it becomes possible to
imagine voters distributed between districts so that the two parties waste about the same proportion of their votes.  The lopsided election problem disappears entirely in the limit as $n\rightarrow\infty$, which corresponds to graphing $\REG_1$ over the domain $\mP$ as done in Figure~\ref{F:graph}.  $\REG_1$ attains all valued between $-.5$ and $.5$ along every ``fixed $\vm$'' line in $\mP$.

The other natural choice is $\lambda=2$:
\begin{equation}\label{E:REG2}\REG_2 = \frac{\sm - (2-C)\vm}{2\left(\frac 12 +\vm\right)\left(\frac 12 - \vm\right)}.\end{equation}
In particular:
\begin{equation}\label{REG2}\REG_2=0 \iff \sm = (2-C)\vm,\end{equation}
which elegantly interpolates between proportionality (in the maximally uncompetitive $C=1$ extreme) and the winner's double bonus principal (in the maximally competitive $C=0$ extreme).  So the winner gets an up-to-double bonus, but only when the districts are competitive enough that the outcome could be attributed in part to random luck.

Figure~\ref{F:graph2} shows the graph of $\REG_2$.  The dependence on $C$ makes it a multi-valued function over $\mP$; the heights of the silver bottom graph and the gold top graph over a point $(\vm,\sm)\in\mP$ are respectively its infimum and supremum among all elections with that outcome.  A section of the gold graph is cut away to make the silver graph below it visible.

The cyan line in the figure shows something apparent from Equation~\ref{E:REG2}: along the slice $\vm=0$, $\REG_2$ is single-valued (which means the infimum equals the supremum), and is a linear function of $\sm$.

\begin{figure}[ht!]
   \scalebox{.20}{\includegraphics{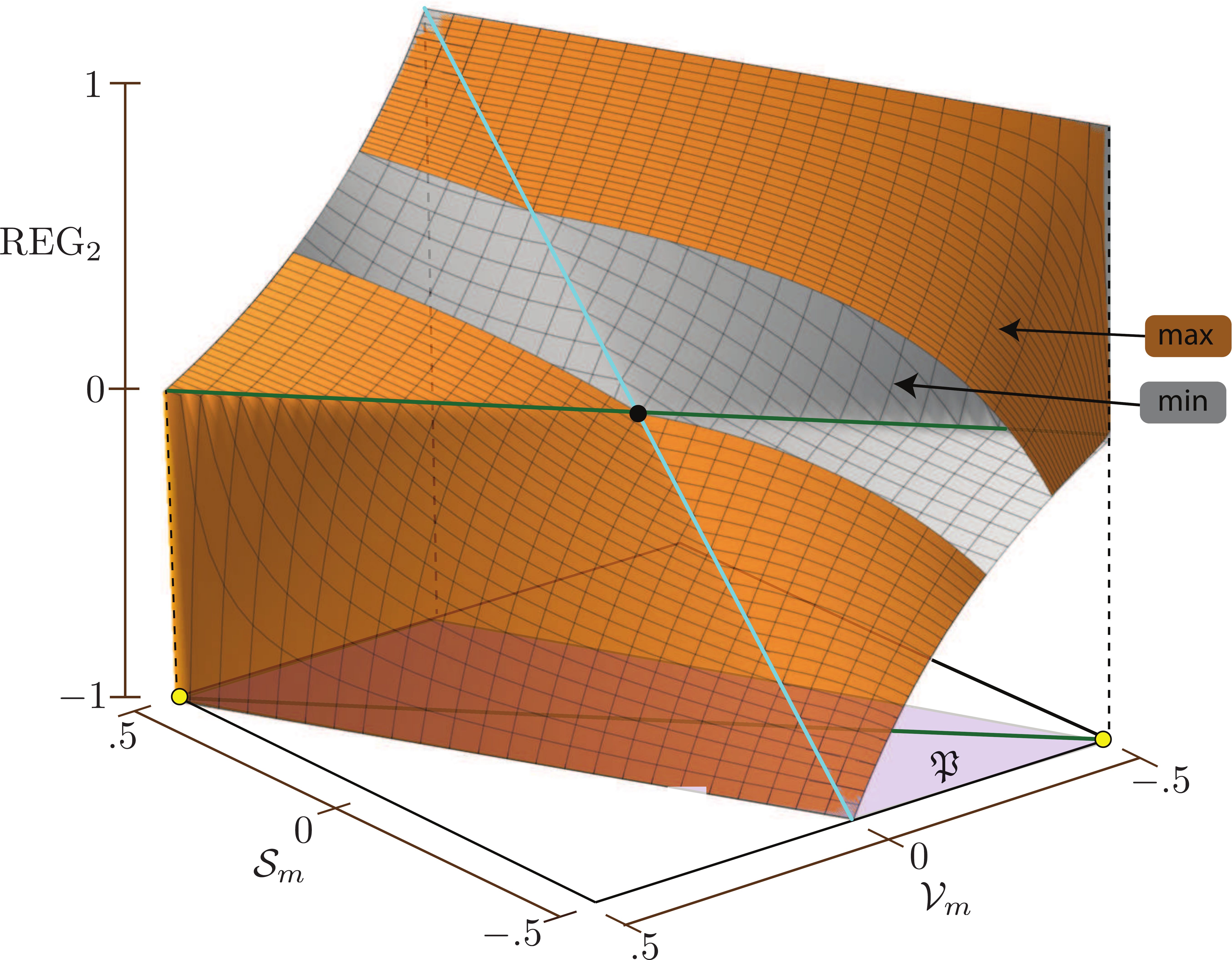}}
\caption{the graph of $\REG_2$ over the domain $\mP$}\label{F:graph2}
   \end{figure}

In fact $\REG_2$ is also single-valued over the boundary of $\mP$, which is why the gold and silver graphs seem to join together along their edges.  To verifying this (and also to create Figure~\ref{F:graph2}) one requires the following technical lemma about the range of possible $C$-values corresponding to any point $(\vm,\sm)\in\mP$, whose proof we leave to the reader:

\begin{lem} For any fixed $\sm\in\left[-\frac 12,\frac 12\right]$, the pair $(\vm,C)$ lies in the tilted rectangle pictured in Figure~\ref{F:Cpar}.
\end{lem}

\begin{figure}[ht!]
   \scalebox{.19}{\includegraphics{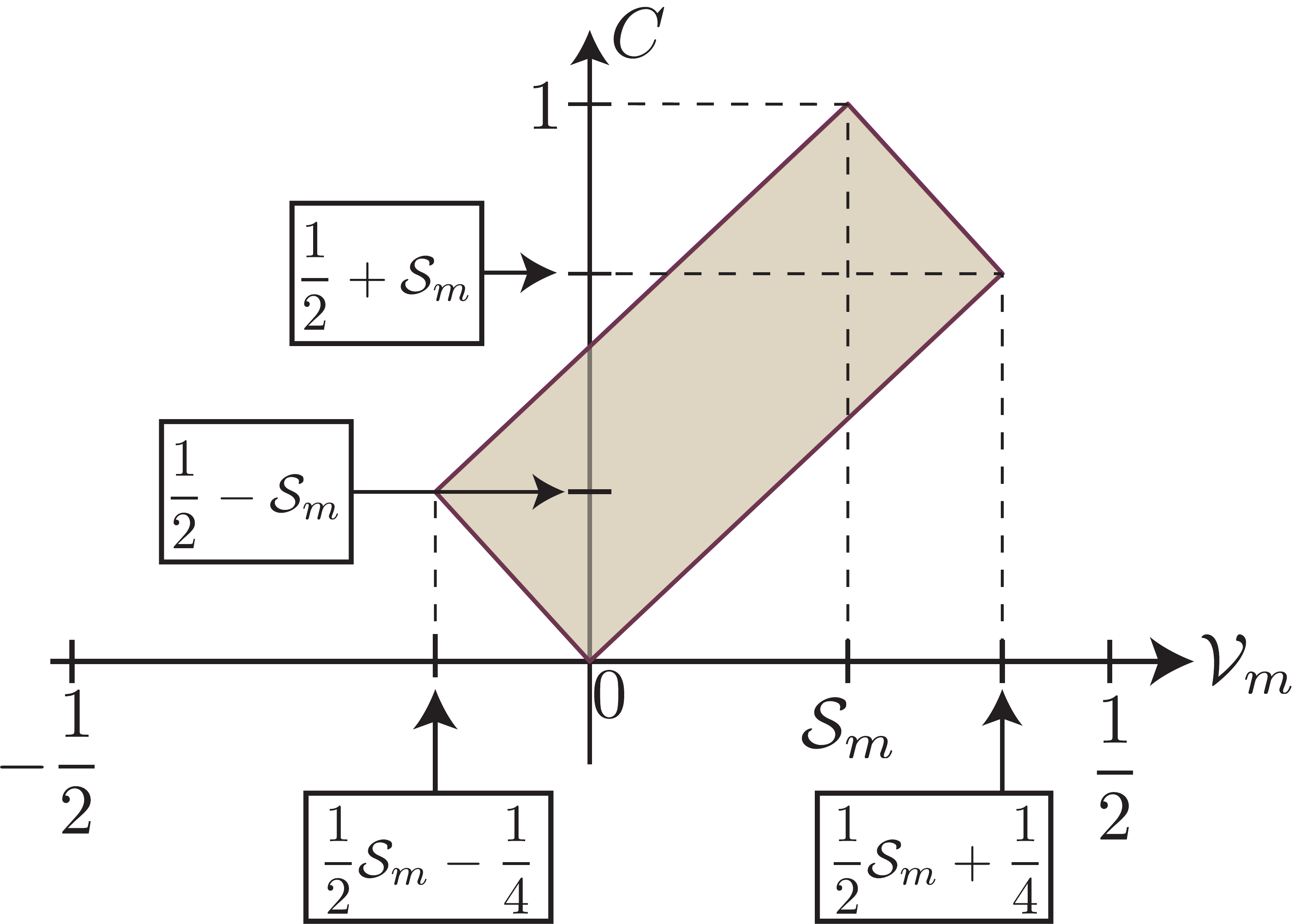}}
\caption{The constrains on $\{\vm,\sm,C\}$}\label{F:Cpar}
   \end{figure}


\section{Comparing the measurements as functions on $\mP$}
Among the gerrymander detection measurements considered in this paper, the most important are: $2\cdot\EG$, $\REG_1$ and $\REG_2$ (here we doubled $\EG$ so that all three measurements have the same range $[-1,1]$).  In this section, we illustrate the differences between these measurements considered as functions on $\mP$.

First we highlight what they have in common: the three measurements agree when the voters are split $50$-$50$ between the two parties (when $\vm=0$), as illustrated by the cyan line in Figures~\ref{F:graph} and~\ref{F:graph2}.  In Wisconsin, $\vm$ was close enough to zero that the three methods nearly agree.

The methods disagree in lopsided situations; the limit version of this statement is the observation that they disagree along the boundary of $\mP$.  The measurements' values along the edges of $\mP$ are given in Figure~\ref{F:boundary}.  In our opinion, $\EG$ and $\REG_2$ get it right along the left and right edges: both equal $1$ along the left edge (where party A has the most seats possible for their given vote share) and they equal $-1$ along the right edge (where they have the fewest).

The yellow points $(\vm,\sm)=\pm(.5,.5)$ in Figure~\ref{F:boundary} are non-removable singularities for $\REG_1$ and $\REG_2$ because the limits along the horizontal and vertical edges don't match up at these points.  Because of this feature, they give more reasonable results than $\EG$ near the yellow points.  Consider outcomes near $(.5,.5)$, which means party A receives almost all the votes and wins almost all the seats.  In such a case, $2\cdot \EG\cong-1$ accuses the minority blue party of manipulating the map, $\REG_1\in[-.5,.5]$ is willing to accuse either party of cheating, and $\REG_2\in[-1,0]$ accuses either the blue party or nobody of cheating.  It might not be clear what's fair here, but $\EG$ clearly gets it wrong.

\begin{figure}[ht!]
   \scalebox{.17}{\includegraphics{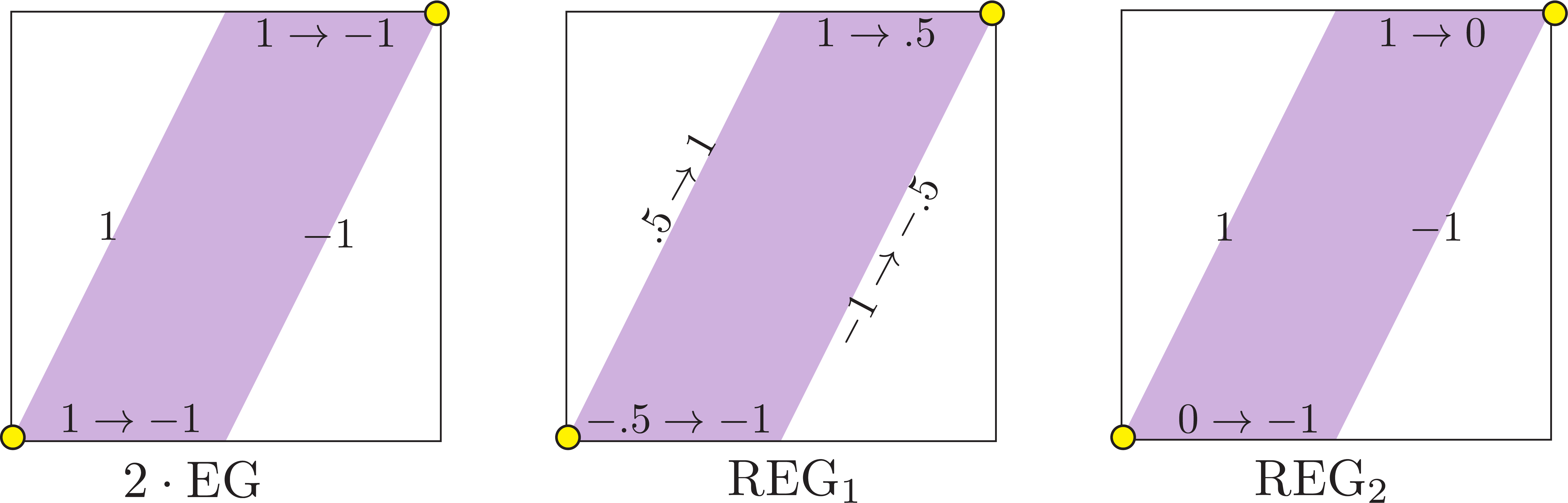}}
\caption{Measurement values along the edges of $\mP$}\label{F:boundary}
   \end{figure}

In practice what matters is not the extreme cases, but rather the answer to the question: which outcomes are considered gerrymandered?  McGhee and Stephanopolous proposed the cut-off $|2\cdot\EG|>.16$, which we believe is equally reasonable for $\REG_1$ and $\REG_2$.  Figure~\ref{F:compare} shows the set of outcomes that are considered ``allowable'' (not gerrymandered) by each of the three measurements.  The measurement $2\cdot\EG$ allows outcomes between the two red lines, while $\REG_1$ allows outcomes between the two purple curves.  The measurement $\REG_2$ \emph{definitely} allows outcomes between the two green regions, and \emph{possibly} allows outcomes in the green regions (but only if the election was sufficiently competitive).  The key observation is that $\REG_2$ beautifully interpolates between the other two measurements.  The region deemed allowable by both $\EG$ and $\REG_1$ is \emph{definitely} allowable by $\REG_2$, while most of the region allowed by only one of $\{\EG,\REG_1\}$ is \emph{possibly} allowed by $\REG_2$.

\begin{figure}[ht!]
   \scalebox{.20}{\includegraphics{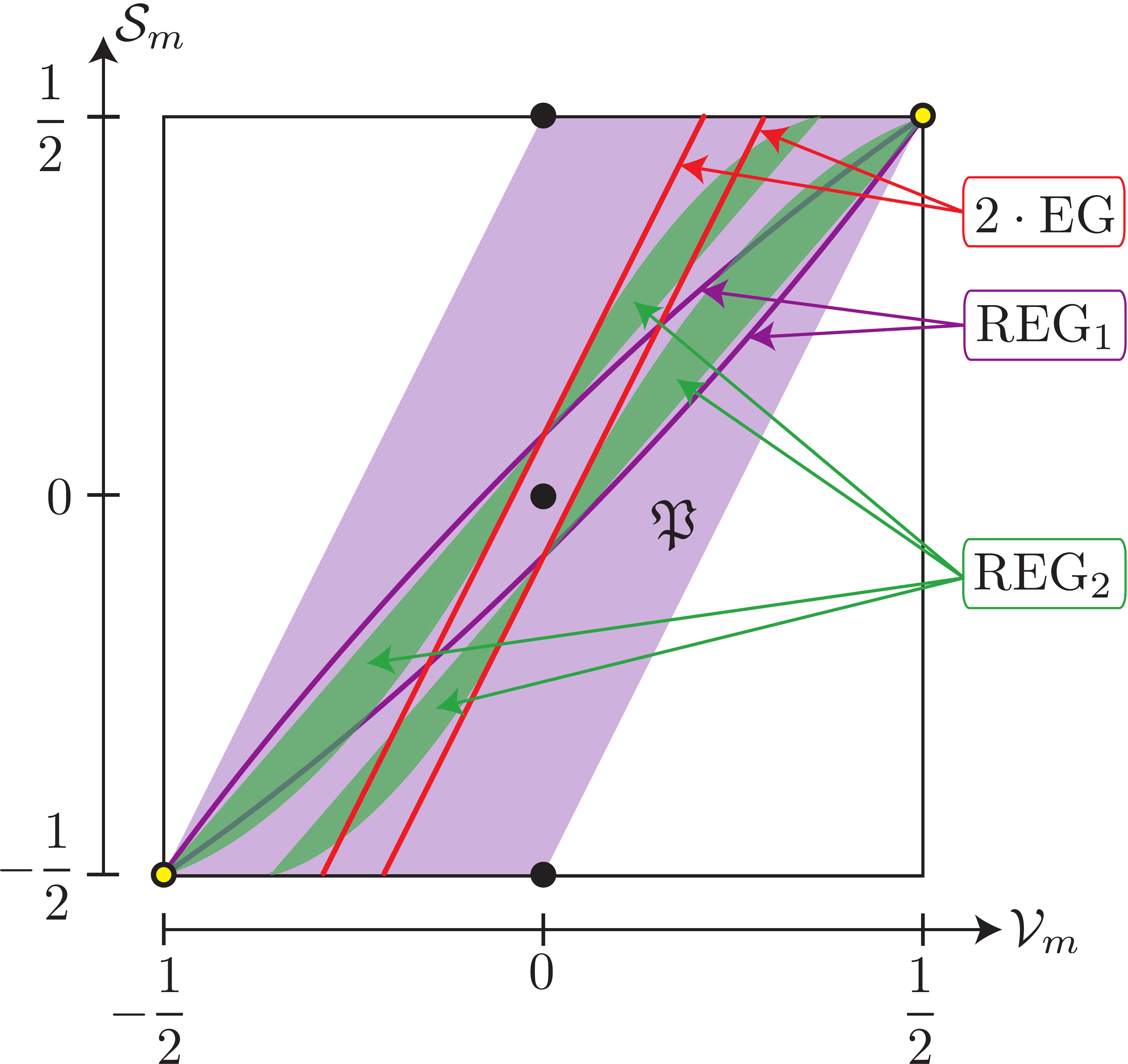}}
\caption{Contours for the cutoffs $\pm.16$}\label{F:compare}
   \end{figure}

\section{Which measurement is best?}
In this section, we enumerate some reasons why we think $\REG_2$ is ultimately preferable to the other two measurements:
\begin{enumerate}
\item It is a significant weakness that $\EG$ is useless in lopsided elections.  In fact we're not even considering $\EG_2$ as a contender because it exacerbates this weakness.  The relative measurements $\REG_1$ and $\REG_2$ both essentially solve this problem.
\item As discussed in the previous section, we think $\REG_2$ has the most reasonable outputs for extreme elections (along the boundary of $\mP$).
\item The result from Section 6 that $\EG$ equals ``weighted differential competitiveness'' would be a strength, but the weighting adds an unpleasant technical complication.  A judge might be more impressed (and less confused) by evidence of an \emph{unweighed} competitiveness gap, especially if this evidence is presented as \emph{independent} from whichever wasted-vote measurement is used.
\item As discussed in Section 6, the choice $\lambda=2$ (rather than $\lambda=1$) is preferable because a district competitiveness cutoff of $1/3$ is more reasonable than $1/2$.  Even though $\REG_1$ (respectively $\REG_2$) is not additive over the districts, it is still true that any district won with $75\%$ (respectively $66\%$) of the vote is neutral in the sense that both parties have the same number of wasted votes.  We consider the $66\%$ option more reasonable.
\item The choice $\lambda=2$ (rather than $\lambda=1$) is preferable because it generalizes correctly to multi-party elections.  If there is only a small ``spoiler'' third party, then a partisan map maker has incentive to carve out districts in which the preferred party has over $50\%$ of the vote (as modeled in~\cite{McGhee2}).  But to incorporate general multi-party elections, including those with many parties of roughly the same size, we believe that the only reasonable definition of an excess vote is: a vote beyond the number received by the second place party in the district.  This definition reduces to $\lambda=2$ when there are only two parties.
\item Even though $\REG_1$ is a nonlinear function of $\vm$ and $\sm$, the fact that if vanishes when $\vm=\sm$ might lead a judge to conclude that it ``just measures the lack of proportionality,''  which would be legally problematic because parties do not have a right to proportional representation.  The measurements $\EG$ and $\REG_2$ are immune to this potential concern.
\item $\EG$ is battle tested.  If one wishes to switch to a relative measurement at this point, $\REG_2$ is preferable to $\REG_1$ because it is a smaller step away from $\EG$.  The illustrations of the previous section show senses in which $\REG_2$ interpolates between $\EG$ and $\REG_1$.
\end{enumerate}
\section{McGhee's efficiency principle}
McGhee compared $\EG$, $\REG_1$ and $\REG_2$ in~\cite{McGhee2}, but he showed much less enthusiasm than us for $\REG_2$, primarily because it failed his \emph{efficiency principle: a gerrymander detection  measurement must increase when party A increases its seat share without any corresponding increase in its vote share.}

If the measurement is a smooth function of $(\vm,\sm)\in\mP$, McGhee's efficiency principle is equivalent to having a positive partial derivative with respect to $\sm$.  The global formulas for $\EG$ and $\REG_1$ have this property, but these formulas were derived assuming the equal turnout hypothesis.  McGhee generalized these formulas to remove this hypothesis, and he showed that the generalizations failed his efficiency principle.  The basic problem is this: a partisan map making team will try to (if possible) carve out small-turnout districts in which their preferred party wins ``at small cost,'' and none of $\EG$, $\REG_1$ or $\REG_2$ would penalize them for doing so.  Veomett studied how extremely $EG$ fails the principle in the presence of malapportionment (unequal voter turnout)~\cite{Veomett}, while McGhee observed in~\cite{McGhee2} that the formula $EG=\sm - 2\cdot\vm$ remains true for malapportioned elections, provided one adopts a reasonable alternative definition of ``excess vote.''

Malapportionment is perhaps a minor issue because a partisan map maker has limited ability to control or take advantage of it.  A more significant issue is that $\REG_2$ fails the efficiency principle even assuming the equal turnout hypothesis.  But we agree with Cover's argument in~\cite{Cover} that this is not necessarily a problem. $\REG_2$ only allows a partisan map making team to improve their seat share (without increasing vote share) if they make the districts more competitive and thereby risk losing seats to bad luck.

To better understand McGhee's efficiency principle, we next prove that it essentially requires the measurement to depend only on $(\vm,\sm)$.  For this, we require a formal definition.

\begin{definition}\label{D:gdm} A \textbf{gerrymander detection measurement} is a function, $\mE$, from the set of all $n$-tuples of pairs of positive integers $$\mathcal{D}=\left((V^A_1,V^B_1),...,(V^A_n,V^B_n)\right)$$ (for all $n\in\mathbb{N}$) to $[-1,1]$, satisfying:
\begin{itemize}
\item[H1:] The indexing of the districts is irrelevant; that is, $\mE(\mD)$ is unchanged when the $n$ elements of $\mD$ are permuted.
\item[H2:] The parties are treated equally; that is, $\mE(\mD)=-\mE(\mD')$ if $\mD'$ is obtained from $\mD$ by swapping the $A$- and $B$-components of each of the $n$ elements.
\item[H3:] Voter scalability; that is, $\mE(m\cdot\mD)=\mE(\mD)$, where $m\cdot\mD$ denote the result of multiplying both the $A$- and $B$- components of all $n$ elements of $\mD$ by the arbitrary positive integer $m$.
\item[H4:] District scalability; that is, for all $m\in\mathbb{N}$, $\mE(\mD)=\mE\underbrace{\left(\mD\cup\cdots\cup\mD\right)}_{\text{$m$ copies}}$, where this $mn$-tuple has $m$ copies of each of the $n$ elements of $\mD$.
\end{itemize}
\end{definition}

This list of properties is very minimal and is satisfied by all measurements considered in this paper.  McGhee's efficiency principle can be precisely formulated as follows: If $\mD_1, \mD_2$ have same number of districts and the same value of $\vm$, but $\mD_2$ has a larger value of $\sm$, then $\mE(\mD_2)>\mE(\mD_1)$.

If $\mE$ is a gerrymander detection measurement, define $\mE_{\text{min}},\mE_{\text{max}}:\mP\rightarrow[-1,1]$ to be like the silver and gold functions graphed in Figure~\ref{F:graph2}.  More precisely, if $(a,b)\in\mP$ is has rational coordinates, then
$$\mE_{\text{min}}(a,b) = \inf\{\mE(\mD)\mid \mD\text{ is such that }\vm=a\text{ and }\sm=b\},$$
while the value of $\mE_{\text{min}}$ on an irrational point is defined as its $\liminf$ on converging sequences of rational points. The function $\mE_{\text{max}}$ is defined analogously.

\begin{prop}\label{P:axium}If a gerrymander detection measurement $\mE$ satisfies McGhee's efficiency principle, then the region between the graphs of $\mE_{\text{min}}$ and $\mE_{\text{max}}$ has zero volume.
\end{prop}
\begin{proof}[Proof sketch] If the region has non-zero volume, then it contains a cube.  From this, it is straightforward to find a pair $\mD_1,\mD_2$ which (after applying H4 so they have the same number of districts) contradict the efficiency principle.
\end{proof}

\begin{example} Let $\mE$ be the symmetry measurement discussed in Section 4, namely twice the height of the simulated seats-votes curve above $(0,0)$ in the $\vm\sm$-plane.  It is straightforward to show that:
$$\mE_{\text{min}}(\vm,\sm)= \begin{cases} 0 &\text{ if $\vm>0$} \\
                                           2\cdot\sm &\text{ if $\vm<0$}\end{cases}
,\quad \mE_{\text{max}}(\vm,\sm)=\begin{cases} 2\cdot\sm &\text{ if $\vm>0$} \\
                                           1 &\text{ if $\vm<0$}\end{cases},$$
because monotonicity is the only general restriction on a seats-votes curve.  Proposition~\ref{P:axium} implies that the measurement fails the efficiency principle.  McGhee exhibited explicit examples of this failure in~\cite{McGhee1} and~\cite{McGhee2}. This measurement does satisfies the much weaker hypothesis that $\mE_{\text{min}}$ and $\mE_{\text{max}}$ individually are monotonically nondecreasing in $\sm$ (as do all measurements considered in this paper).  It would be reasonable to add this unassuming hypothesis to Definition~\ref{D:gdm}
\end{example}

McGhee's principle essentially requires a measurement to be a function of $(\vm,\sm)\in\mP$, except possibly on a set of measure zero.  We consider this very restrictive.  Although it is interesting to ask about the best function on $\mP$, it is also necessary to build a framework that includes more general measurements.  We consider it a benefit (not a drawback) that $\REG_2$ depends on $C$ in addition to $(\vm,\sm)$.  Another important measurement that fails the principle is $\mE=C^B-C^A$ (unweighted differential competitiveness) which could be used in tandem with $\REG_2$ as independent supporting evidence.

\bibliographystyle{amsplain}

\end{document}